\documentclass[10pt,twocolumn,twoside]{IEEEtran}

\usepackage{graphicx,epsfig}
\usepackage{dblfloatfix}
\usepackage{fixltx2e}
\usepackage{float}

\usepackage{cite}
\usepackage{amsmath}
\usepackage{amsfonts}
\usepackage{enumerate}

\newtheorem{theorem}{Theorem}

\newtheorem{algo}[theorem]{Algorithm}

\newtheorem{lemma}[theorem]{Lemma}

\newtheorem{proposition}[theorem]{Proposition}

\newcounter{proposition}
\begin{document}

\title{Convergence and Density Evolution of a Low-Complexity MIMO Detector based on Forward-Backward Recursion over a Ring}

\author{Seokhyun Yoon,~\IEEEmembership{Member,~IEEE} \\
\thanks{Copyright (c) 2013 IEEE. Personal use of this material is permitted. However, permission to use this material for any other purposes must be obtained from the IEEE by sending a request to pubs-permissions@ieee.org. } 
\thanks{This work was supported by Basic Science Research Program through the National Research Foundation of Korea (NRF) funded by the Ministry of Education, Science and Technology (NRF-2012R1A1A2038807). }
\thanks{S. Yoon is with the Department of Electronics Engineering, Dankook University, Korea (e-mail: syoon@dku.edu)} }

\maketitle \setcounter{page}{1} 
%
%
%

\markboth{\emph{submitted to IEEE Trans. on Vehicular Technology}, 2015}%
{Belief Propagation over Ring-Type pair-wise Graph}

\begin{abstract}
Convergence and density evolution of a low complexity, iterative MIMO detection based on belief propagation (BP) over a ring-type pair-wise graph are presented in this paper. The detection algorithm to be considered is effectively a forward-backward recursion and was originally proposed in [13], where the link level performance and the convergence for Gaussian input were analyzed. Presented here are the convergence proof for discrete alphabet and the density evolution framework for binary input to give an asymptotic performance in terms of average SINR and bit error rate (BER) without channel coding. The BER curve obtained via density evolution shows a good match with simulation results, verifying the effectiveness of the density evolution analysis and the performance of the detection algorithm.
\end{abstract}




\section{Introduction}

During the last decade, there were lots of works on belief propagation based MIMO detection, in terms of detection in multi-antenna spatial multiplexing or multiuser detection in code-division multiplexing \cite{R01,R02,R03,R04,R05,R06}. In these approaches, the MIMO channel is modeled as a fully-connected factor graph, which consists of a multiple $N$ factor nodes representing the received signal, a multiple $M$ variable nodes representing the hidden data, and the edges connecting the factor nodes with the variable nodes. The resulting graph has maximal edge degree, i.e., every factor node is connected to every variable node. 

In terms of performance, \cite{R05} and \cite{R06} showed that BP asymptotically performs the same as maximum {\it a posteriori} (MAP) detector, if the graphical model is sparse enough. Especially, \cite{R05} showed that BP performs the same as MAP even if the graph is dense while the system load (which, in our context, is the multiplexing order normalized to the number of transmit antenna) is less than a certain limit. 

In terms of complexity, however, the complexity of BP based detection over the fully connected factor graph is as high as MAP detector due mainly to the marginalization operation required for the message update at the factor nodes. To reduce the computational complexity, model simplification approaches have been studied. Especially, in \cite{R07}, it was suggested to prune some edges in the fully connected factor graph, based on the strength of the channel coefficients, i.e., to prune edges corresponding to those variable-factor node pairs with small value of $|h_{jk}|$. By using only $d_f < M$ edges per factor node (i.e., pruning $M-d_f$ edges), the complexity is reduced by a factor of $1/2^{m(M-d_f)}$ relative to MAP of complexity $O(2^{mM})$. The problem of this scheme is that $d_f$ must not be too small to ensure reasonable performance. 

Other interesting graph-based approaches are those in \cite{R08,R09,R10,R11,R12,R13} based on the pair-wise Markov random field (MRF). In MRF, we have nodes representing the hidden data and the edges reflecting the local dependency among them. The local dependency is represented by potential functions and, specifically, in pair-wise MRF they are functions of one or two variables. In fact, as noticed in \cite{R14,R15,R16,R17}, a multivariate Gaussian function can be decomposed into a product of functions of one or two variables resulting in a fully connected pair-wise MRF. Unfortunately, however, BP over the pair-wise MRF based on this potential function does not work well for higher-order modulation, such as 16QAM. 

To overcome such problem, \cite{R12} and \cite{R13} considered using potential functions obtained by a linear transformation, e.g., by a conditional MMSE estimator. Leaving only two variables, one can construct pair-wise graphical model resulting in a low complexity detection algorithm. As a matter of fact, the edge pruning in \cite{R08}, \cite{R09} and \cite{R12,R13} are special case of the channel truncation approach in  \cite{R20}, either in zero-forcing sense or in MMSE sense. Similar to sphere decoding, these detectors are a two stage detector, where the channel is first truncated (pruning edges) to simplify the graph and, then, post-joint detection is performed as a BP over the simplified graphical model.

Gaussian BP, as those in \cite{R16,R17,R18}, can also be considered for low complexity MIMO detection. In Gaussian BP, the input data and messages are all assumed to be Gaussian so that the message and posterior probability (belief) can be represented by a pair of mean and variance, resulting in a very simple message update rule. As shown in \cite{R16} and \cite{R18}, (and also in \cite{R13}), however, the algorithm converges only to the linear minimum mean squared error (LMMSE) solution that is inferior to the MAP detector for non-Gaussian input. 

One lesson from \cite{R13} is that  the BP based MIMO detector over the ring-type pair-wise graph in \cite{R13} is always convergent regardless of its alphabet size. Note that the BP over fully-connected pair-wise model in \cite{R10}, \cite{R11}, and \cite{R13} do not converge, especially when alphabet size is finite and higher than 2. The guarantee of the convergence of the BP over ring-type model in \cite{R13} might come from avoiding short loop. This result is consistent with the results in \cite{R21} and \cite{R22}, where it was shown that BP over a graphical model with a single cycle always converges. According to \cite{R05} and \cite{R06}, as well as more recent simulation results in \cite{R09} and \cite{R13}, however, limiting edge degree and keeping graphical model sparse seems to be a must for successful convergence of BP based detection, especially for use of higher order modulation. 

In this paper, we extend the convergence of the iterative MIMO detector based on the ring-type pair-wise graph to the discrete alphabet. We also develop density evolution framework to characterize the stationary SINR distribution after convergence, which will give us a deep insight into the mechanism and the performance of the algorithm.

The paper is organized as follows. In the next section, the previous works are briefly introduced for the development of the analysis in the subsequent sections. In section III, the convergence proof is provided for discrete alphabet and, in section IV, the density evolution approach is developed for SINR analysis of the MIMO detection algorithm under consideration. Some numerical results are given in section V to validates the density evolution approach for SINR analysis and, finally, the concluding remarks in section VI.

\section{System model and previous works}
\label{Sec:II}

\subsection{System Model}

A Gaussian MIMO system with an $N\times M$ channel matrix ${\pmb H} (N \geq M)$ is modeled as
\begin{align}
{\pmb y} = {\pmb H} {\pmb x} + {\pmb n} = \sum_{k=1}^M {\pmb h}_k x_k + {\pmb n} \label{eq01}
\end{align}
where ${\pmb x}$ is an $M \times 1$ transmitted data symbol vector, ${\pmb n}$ is an $N \times 1$ noise vector, ${\pmb y}$ is an $N \times 1$ received signal vector and ${\pmb h}_m$ is the $m$th column of ${\pmb H}$. The noise vector ${\pmb n}$ is assumed to be complex Gaussian with mean ${\pmb 0}$ and covariance ${\mathbb E}[{\pmb n} {\pmb n}^H] = \sigma^2 {\pmb I}$ and the transmitted data symbol vector ${\pmb x}$ is assumed to have mean ${\pmb 0}$ and covariance matrix ${\mathbb E}[{\pmb x} {\pmb x}^H] = {\pmb I}$, where ${\mathbb E}(\cdot)$ denotes expectation. In practice, each element of ${\pmb x}$ is drawn from a finite alphabet set $\Xi$ of size $2^m$, such as QPSK and 16-QAM, for which $m$ = 2 and 4, respectively.


\begin{figure}[!t]
   \centerline{\resizebox{0.5\columnwidth}{!}{\includegraphics{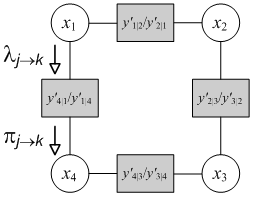}}}
   \caption{The ring-type pair-wise graph for $4 \times N$ MIMO channel}
   \label{Fig01}
\end{figure}

\subsection{Low complexity detection based on BP over ring-type pair-wise graph \cite{R13}}

Our start point is the one in \cite{R13}, especially the BP over ring-type pair-wise graph. The graphical model is shown in Fig.1, over which the BP algorithm is effectively a forward-backward recursion and can be summarized as follows.\newline

\noindent \textsf{BP over ring-type pair-wise graph} 

\noindent \texttt{{\small Given the messages in the previous iteration, $\pi _{(j \mp 1)_M \to j} (x_j )$, they are recursively updated for all $j$ as}}
\begin{align}
\pi_{j \to (j \pm 1)_M} &(x_{(j \pm 1)_M}=s) \nonumber \\ = & \sum_{s'  \in \Xi } 
{\gamma_{(j \pm 1)_M|j}(s |s') \cdot \pi _{(j \mp 1)_M \to j} (x_{j}=s' )}\label{eq02}
\end{align}
\texttt{{\small with $\gamma_{j|{i}}(s |s' )$ given by (\ref{eq09}).\newline  \indent After a pre-defined number of iterations, the belief is finally obtained by}}
\begin{align}
b(x_j ) = \pi _{(j + 1)_M \to j} (x_j ) \cdot \pi _{(j - 1)_M \to j} (x_j ). \label{eq03}
\end{align}
\newline
In (\ref{eq02}) and (\ref{eq03}), $(\cdot)_M$ is one-base modulo-$M$ operation. Since using $(\cdot)_M$ is cumbersome, it will later be omitted.

In this algorithm, we used only factor to variable node message, $\pi _{i \to j} (x_j )$, since there is only two factor nodes connected to a variable node such that variable nodes simply pass the incoming message to the opposite side. 

In (\ref{eq02}), the translation function, $\gamma_{j|{j \pm 1}}(s |s' )$, (also known as branch metric in the context of forward-backward recursion along a trellis) is based on the conditional MMSE estimator of $x_j$ given $x_{j\pm1}$. By defining the conditional estimator of $x_j$ given $x_i$ as 
\begin{align}
{\pmb{w}}_{j|i}^{}  = {\pmb{K}}_{\{ j,i\} }^{ - 1} {\pmb{h}}_j \label{eq04}
\end{align}
and applying it to the received signal vector ${\pmb y}$, we have
\begin{align}
y_{j|i}  = {\pmb{w}}_{j|i}^H {\pmb{y}} = a_{j|i,j} x_j  + a_{j|i,i} x_i  + n_{j|i} \label{eq04a}
\end{align}
where
\begin{align}
{\pmb{K}}_\Phi ^{}  &= \sigma ^2 {\pmb{I}} + \sum\nolimits_{k \notin \Phi }^{} {{\pmb{h}}_k^{} {\pmb{h}}_k^H } \label{eq05} \\
a_{j|i,k}  &= {\pmb{w}}_{j|i}^H {\pmb{h}}_k  = {\pmb{h}}_j^H {\pmb{K}}_{\{ j,i\} }^{ - 1} {\pmb{h}}_k^{}\hspace{8pt} \textrm{for}\; k = i\; \textrm{or}\; j \label{eq06} \\
\mathbb{E}|{n_{j|i}}{|^2} &= {\pmb{w}}_{j|i}^H{{\pmb{K}}_{\{ j,i\} }}{\pmb{c}}_{j|i}^{} = {\pmb{h}}_j^H{\pmb{K}}_{\{ j,i\} }^{ - 1}{\pmb{h}}_j^{} \equiv \sigma _{j|i}^2. \label{eq07}
\end{align}
Note that $\sigma _{j|i}^2 = a_{j|i,j}$  and the parameters from (\ref{eq05}) to (\ref{eq07}) are computed from the channel parameters, $\pmb{H}$ and $\sigma^2$, and the received signal $\pmb{y}$. In the truncated signal model in (\ref{eq04a}), we assume the noise + interference, $n_{j|i}$, to be a Gaussian, from which the translation function is given by
\begin{align}
\gamma_{{j}|i}&(x_{j} |x_i ) \propto  
\mathcal{CN}(y_{j|i} ;a_{j|i,j} x_j  + a_{j|i,i} x_i ,\sigma _{j|i}^2 ) 
\label{eq09}
\end{align}
where $\mathcal{CN}(x;a,b)$ is the complex Gaussian density function with mean $a$ and variance $b$ given by
\begin{align}
\mathcal{CN}(y;\mu ,\sigma^2 ) & \equiv 
\frac{1}{  \pi \sigma^2  }\exp \left( - \frac{|{y} - {\mu}|^2}{\sigma^2}  \right) \nonumber
\end{align}

The translation function given by Gaussian density function in (\ref{eq09}) assumes data symbols other than $x_j$ and $x_i$ are all Gaussian. This assumption, however, is only for pruning edges and to simplify the graphical model as shown in Fig.1 while the messages themselves in (\ref{eq02}) are not necessarily Gaussian too. 

On the other hand, when the input is indeed Gaussian, the forward backward algorithm in (\ref{eq02}) reduces to an update rule for mean and variance \cite{R13} as follows. \newline

\noindent \textsf{Gaussian BP over ring-type pair-wise graph} 

\noindent\texttt{\small Given the messages in the previous iteration, 
$(\mu_{\pi,j \mp1 \rightarrow j},\sigma^{2}_{\pi,j \mp1 \rightarrow j})$ $\forall j$,
they are recursively updat- ed by}
\begin{align}
\sigma _{\pi ,j \to j \pm 1}^2  &= \frac{1}{{1 + \sigma _{j \pm 1|j}^2 }} + \frac{{|a_{j \pm 1|j,j} |^2 }}{{(1 + \sigma _{j \pm 1|j}^2 )^2 }} \cdot \sigma _{\pi ,j \mp 1 \to j}^{ 2} \label{eq10} \\
\mu _{\pi ,j \to j \pm 1}  &= \frac{y_{j \pm 1|j}}{{1 + \sigma _{j \pm 1|j}^2 }}  - \frac{{a_{j \pm 1|j,j} }}{{1 + \sigma _{j \pm 1|j}^2 }} \cdot \mu _{\pi ,j \mp 1 \to j}^{} \label{eq11}
\end{align}
\newline

In \cite{R13}, it is proved that the mean and variance in the Gaussian forward-backward recursion converge respectively to MMSE estimates and its corresponding MMSE, i.e., 
\begin{align}
\mu _{\pi ,j \mp 1 \to j} \rightarrow &\; \hat{x}_j = {\pmb h}_j^H {\pmb K}^{-1} {\pmb y}  \label{eq12} \\
\sigma^{2}_{\pi,j \mp1 \rightarrow j} \rightarrow &\; \text{MMSE}_j = 1- {\pmb h}_j^H {\pmb K}^{-1} {\pmb h}_j. \label{eq13}
\end{align}
as the number of iteration goes to infinity. Since MAP detector becomes linear MMSE estimator for the Gaussian input, it shows the optimality of the scheme for Gaussian input. However, for the non-Gaussian input, MMSE estimator is far inferior to MAP detector.

	In the rest of this paper, we deal with the convergence of the forward-backward algorithm in (\ref{eq02}) for arbitrary discrete alphabet and its density evolution characteristic for binary input, which will give us a deeper insight into its convergence and performance.

\section{Convergence for Discrete Alphabet} \label{Sec:II}

In this section, we provide two convergence proofs of the BP over ring-type pair-wise graph, one for arbitrary discrete alphabet and the other for binary one. The signal model for binary input provides framework not only for the convergence proof but also for density evolution analysis to be discussed in the next section.

\subsection{Convergence proof for arbitrary discrete input } 

The convergence proof in this subsection is based on the 'Perron-Frobenius theorem'. Although it has already been discussed in \cite{R22}, we provide it in our context here for the reader's convenience. 

Suppose that data symbols are drawn from an $2^m$-ary alphabet set $\Xi = \{ s_1, s_2, ..., s_{2^m} \}$. The forward-backward recursion in (\ref{eq02}) can be concisely expressed as
\begin{align}
{\pmb \pi}_{j \to j \pm 1} &= \frac{1}{\alpha_{j\pm1}} {\pmb A}_{j \pm 1|j} {\pmb \pi}_{j\mp 1 \to j} \label{eqa1}
\end{align}
where ${\pmb \pi}_{i \to j}$ is $M \times 1$ message vector of which the $m$th element is given by ${\pi}_{i \to j}(x_j = s_m)$, ${\pmb A}_{j|i}$ is $M \times M$ translation matrix of which the ($m,n$)th element is $\gamma_{j|i}(x_j=s_m|x_i=s_n)$ and $\alpha_j$ is the normalization constant, such that $\| {\pmb \pi}_{i \to j} \|_1 = 1$, i.e., $\alpha_j = \| {\pmb A}_{j|i} {\pmb \pi}_{i \to j} \|_1$.{\footnote{$\| {\pmb a} \|_1$ denotes the $L_1$-norm of a vector ${\pmb a}$}} Note that all element of ${\pmb A}_{j|i}$ are positive real. 

Define ${\pmb F}_j$ and ${\pmb B}_j$ as the translation matrix for one complete turn of forward and backward recursion, respectively, i.e.,
\begin{align}
{\pmb F}_j &=  {\pmb A}_{j|j-1} {\pmb A}_{j-1|j-2} \cdots {\pmb A}_{2|1} {\pmb A}_{1|M} \cdots {\pmb A}_{j+1|j} \label{eqa2} \\
{\pmb B}_j &=  {\pmb A}_{j|j+1} {\pmb A}_{j+1|j+2} \cdots {\pmb A}_{M-1|M} {\pmb A}_{M|1} \cdots {\pmb A}_{j-1|j} \label{eqa3} 
\end{align}
Then, the message vector at the $k$th turn can be expressed as
\begin{align}
{\pmb \pi}_{j-1 \to j}^{(k)} \propto {\pmb F}_j  {\pmb \pi}_{j-1 \to j}^{(k-1)} = {\pmb F}_j^k  {\pmb \pi}_{0}\label{eqa4} \\
{\pmb \pi}_{j+1 \to j}^{(k)} \propto {\pmb B}_j  {\pmb \pi}_{j+1 \to j}^{(k-1)} = {\pmb B}_j^k  {\pmb \pi}_{0}\label{eqa5} \end{align}
where ${\pmb \pi}_{j \pm 1 \to j}^{(k)}$ is the message at the $k$th turn and ${\pmb \pi}_{0}$ is the initial message, which is typically set to uniform distribution. Now, one can prove the following theorem.

\begin{theorem}
With any initial message, ${\pmb \pi}_{0}$, the message ${\pmb \pi}_{j - 1 \to j}^{(k)}$ in (\ref{eqa4}) converges (up to a normalization constant) to the eigenvector of ${\pmb F}_j$ corresponding to its largest eigenvalue. Similarly, ${\pmb \pi}_{j + 1 \to j}^{(k)}$ in (\ref{eqa5}) converges to the eigenvector of ${\pmb B}_j$ corresponding to its largest eigenvalue.
\end{theorem}

\begin{proof}
First, we decomposed the transition matrix for one complete turn, ${\pmb F}_j$, into
\begin{align}
{\pmb F}_j &= {\pmb E}_j {\pmb D}_j {\pmb E}_j^{-1} \nonumber 
\end{align}
where ${\pmb D}_j$ is diagonal matrix with its $m$th diagonal element is the $m$th eigenvalue, $\lambda_m$, and ${\pmb E}_j$ is the eigenbasis, of which the $m$th column, ${\pmb e}_m$, is the eigenvector for $\lambda_m$. Then, we have
\begin{align}
{\pmb F}_j^{k} &= {\pmb E}_j {\pmb D}_j^{k} {\pmb E}_j^{-1} \nonumber \\ 
&= \left[ \lambda_1^k {\pmb e}_1  \lambda_2^k {\pmb e}_2  \cdots \lambda_M^k {\pmb e}_M  \right] \cdot 
\begin{bmatrix}
i_{11} & i_{12} & \cdots & i_{1M} \\   
i_{21} & i_{22} &            & i_{2M} \\   
\vdots &  & \ddots & \vdots \\   
i_{M1} & i_{M2} & \cdots & i_{MM} \\   
\end{bmatrix} \nonumber
\end{align}
where $i_{mn}$ is the ($m,n$)th element of ${\pmb E}_j^{-1}$. Note that, 
with an initial message, ${\pmb \pi}_0$, we have
\begin{align}
\left[{\pmb \pi}^{(k)}_{j-1 \to j} \right]_m &= \left[ {\pmb F}_j^k {\pmb \pi}_0 \right]_m = \sum_{n=1}^{M} \sum_{l=1}^{M} e_{ml} \lambda_l^k i_{ln} \pi_{0,n} \label{eqa6}
\end{align}
Since all the entries of ${\pmb F}_j^{k}$ are positive, we see, from Perron-Frobenius theorem, that the matrix, ${\pmb F}_j$, has single largest real eigenvalue (a.k.a. Perron-Frobenius eigenvalue) of which the corresponding eigenvector has (or can be chosen to have) all positive entries. Let ${\lambda}_{l^*}$ be the largest real eigenvalue and ${\pmb e}_{l^*}$ the corresponding eigenvector. Then, by taking limit $k \to \infty$ to the normalized message $[{\pmb \pi}^{(k)}_{j-1 \to j}]_m$ $ / \| {\pmb \pi}_{j-1 \to j}^{(k)} \|_1$, we have
\begin{align}
\lim_{k \to \infty} \frac{[{\pmb \pi}^{(k)}_{j-1 \to j} ]_m}{\| {\pmb \pi}_{j-1 \to j}^{(k)} \|_1} &\to \frac{ \sum_{n=1}^{M} e_{ml^*} \lambda_{l^*}^k i_{{l^*}n} \pi_{0,n} }{ \sum_{m'=1}^{M} \sum_{n=1}^{M} e_{m'l^*} \lambda_{l^*}^k i_{{l^*}n} \pi_{0,n} } \nonumber \\ 
&= \frac{ e_{ml^*} \lambda_{l^*}^k \sum_{n=1}^{M} i_{{l^*}n} \pi_{0,n} }{ \sum_{m'=1}^{M}  e_{m'l^*} \lambda_{l^*}^k \sum_{n=1}^{M} i_{{l^*}n} \pi_{0,n} } \nonumber \\
&= \frac{ e_{ml^*}  }{ \sum_{m'=1}^{M}  e_{m'l^*} } \nonumber \\
&= \frac{ [{\pmb e}_{l^*}]_m}{\| {\pmb e}_{l^*} \|_1} \label{eqa7}
\end{align}
where, in the r.h.s. of the first line, we took only the term with the largest eigenvalue since other terms are negligible as $k \to \infty$.

The convergence for the backward recursion can be proved similarly.
\end{proof}


Although the convergence proof below is applicable to arbitrary alphabet, it is not suitable for further SINR and BER analysis. So, we provides another signal model which deals directly with the log-likelihood ratio (LLR) by restricting the data to binary. As will be shown later, it gives us a more convenient way for SINR and BER analysis via density evolution. 

\subsection{Message passing for binary input}

For binary input, the message can be summarized by a scalar, i.e., the log likelihood ratio(LLR). Define the message and {\it a priori} LLR as
\begin{align}
l_{i \rightarrow j} &= \log \frac{\pi_{i \rightarrow j}(x_{j} = +1)}{\pi_{i \rightarrow j}(x_{j} = -1)} \label{eq14} \\
l_{a,i} &= \log \frac{p(x_{i} = +1)}{p(x_{i} = -1)} \label{eq15}
\end{align}
such that
\begin{align}
p(x_{i} = \pm 1) &= \frac{e^{\pm l_{a,i}/2}}{e^{+ l_{a,i}/2}+e^{-l_{a,i}/2}} \nonumber
\end{align}
Then, the forward recursion in (\ref{eq02}), together with (\ref{eq14}), can be expressed as (\ref{eq16}) shown on top of the next page, where we assume uniform priors. 

\begin{figure*}[tp]
\begin{align}
l_{i \rightarrow j} &= \log \left( \frac{ p(x_j=+1) \displaystyle\sum_{x_{j-1}=\pm 1} \mathcal{CN}(y'_{j|j-1} ;+a_{j|j-1,j}  + a_{j|j-1,j-1} x_{j-1} ,\sigma _{j|j-1}^2 ) \pi _{j - 2 \to j-1} (x_{j-1} ) }{p(x_j=-1) \displaystyle\sum_{x_{j-1}=\pm 1} \mathcal{CN}(y'_{j|j-1} ;-a_{j|j-1,j}  + a_{j|j-1,j-1} x_{j-1} ,\sigma _{j|j-1}^2 ) \pi _{j - 2 \to j-1} (x_{j-1} )} \right) \nonumber \\
&= l_{a,j} + 
\log \left(  \frac{ \displaystyle\sum_{x_{j-1}=\pm 1} \exp \left( -\frac{|y_{j|j-1}-a_{j|j-1,j}-a_{j|j-1,j-1}x_{j-1}|^2}{\sigma^2_{j|j-1}} + \frac{l_{j-2 \rightarrow j-1}x_{j-1}}{2} \right) }
{ \displaystyle\sum_{x_{j-1}=\pm 1} \exp \left( -\frac{|y_{j|j-1}+a_{j|j-1,j}-a_{j|j-1,j-1}x_{j-1}|^2}{\sigma^2_{j|j-1}} + \frac{l_{j-2 \rightarrow j-1}x_{j-1}}{2} \right) } \right)  \label{eq16} 
\end{align}
\hrulefill
\end{figure*}

Removing common terms in the numerator and denominator in (\ref{eq16}) and defining a function $\zeta(x;c)$ of $x$ with a parameter $c$ as
\begin{align}
\zeta(x;c) = -\log \left( \frac{ e^{x/2-c} + e^{-x/2+c}}{e^{x/2+c} + e^{-x/2-c}} \right) \label{eq17} 
\end{align}
(\ref{eq16}) can be concisely rewritten as
\begin{align}
l_{j-1 \to j} = l_{a,j} + 4y^{(R)}_{j|j-1}-\zeta(l_{j-2 \to j-1} + 2d_{j|j-1} ; c_{j|j-1}  ) \label{eq18} 
\end{align}
\noindent where
\begin{align}
y^{(R)}_{j|i} &= \Re[y_{j|i}] \label{eq19} \\ 
c_{j|i} &= \frac{2}{\sigma^2_{j|i}} \Re[a^{*}_{j|i,j}a_{j|i,i}] = 2\Re[a_{j|i,i}] = 2a^{(R)}_{j|i,i} \label{eq20} 
\end{align}
\begin{align}
d_{j|i} &= \frac{2}{\sigma^2_{j|i}} \Re[a^{*}_{j|i,i}y_{j|i}] \nonumber \\
&= 2\Re[a^{*}_{j|i,i}] x_j + \frac{2|a_{j|i,i}|^2}{\sigma^2_{j|i}} x_i + \frac{2}{\sigma^2_{j|i}} \Re[a^{*}_{j|i,j}n_{j|i}] \nonumber \\
&= 2a^{(R)}_{j|i,i} x_j + \frac{2|a_{j|i,i}|^2}{\sigma^2_{j|i}} x_i \nonumber \\
&\;\;\;\; + \frac{2}{\sigma^2_{j|i}} \left(  a^{(R)}_{j|i,j}n^{(R)}_{j|i} + a^{(I)}_{j|i,j}n^{(I)}_{j|i} \right) \label{eq21}
\end{align}
where we denote the real and imaginary part of a complex variable as superscript ($R$) and ($I$), respectively, for notational simplicity. 



The non-linear function $\zeta(x;c)$ in (\ref{eq17}) is a monotonic function of $x$, either increasing if $c > 0$ or decreasing if $c < 0$, and has the following properties. 
\begin{align}
&(a) \;\; \zeta'(x;c) = \frac{\mathrm{d}}{\mathrm{d}x} \zeta(x;c) \nonumber \\
& \;\;\;\;\;\;\;\;\;\;\;\;\;\;\;\;\; = \frac{1}{2}\tanh \left( \frac{x}{2} +c \right) - \frac{1}{2}\tanh \left( \frac{x}{2} -c \right) \\
&(b) \;\; |\zeta'(x;c)| < 1 \; \forall x \\
&(c) \;\; \lim_{x \to \infty}\zeta(x;c) \to \pm2c \; {\textrm {(saturation)}}. \\
&(d) \;\; \lim_{c \to 0} \frac{\zeta(x;c)}{2c}  \to \tanh \left( \frac{x}{2} \right). 
\end{align}
Note that  
\begin{align}
& \frac{ \zeta( l_{j-2 \to j-1} + 2d_{j|j-1} ; c_{j|j-1} ) }{2c_{j|j-1}} 
\approx \tanh \left( \frac{ l_{j-2 \to j-1} }{2} + d_{j|j-1} \right)  \nonumber
\end{align}
can be regarded as an estimate of $x_{j-1}$ based on the information provided from the current observation, $y_{j|j-1}$, and the message from the previous node, $l_{j-2 \to j-1}$, through the forward-backward recursion. 

\begin{figure}[!t]
  \centerline{\resizebox{1.1\columnwidth}{!}{\includegraphics{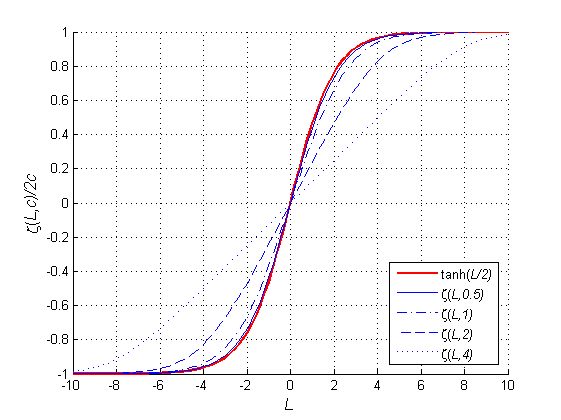}}}
   \caption{Plot of the function $\zeta(x;c)/2c$ with $c$ = 0.5, 1, 2, and 4. With $c < 0.5$, $\zeta(x;c)/2c$ can be approximated to  $-\tanh (x/2)$.}
   \label{Fig02}
\end{figure}

\begin{figure}[!t]
  \centerline{\resizebox{0.7\columnwidth}{!}{\includegraphics{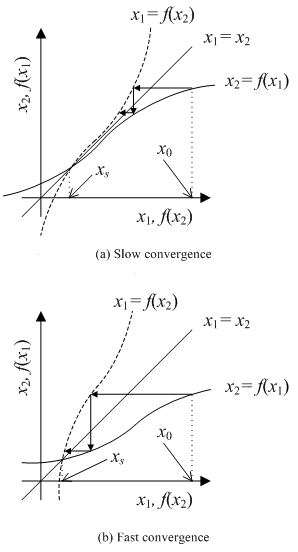}}}
   \caption{Two examples showing the convergence in {\it lemma 1}. The condition $f'(x)|<1$ ensures the convergence to $x_s$, regardless of its convergence speed.}
   \label{Fig02a}
\end{figure}

\subsection{Convergence proof for binary input}

To prove the convergence of the forward-backward recursion in (\ref{eq18}) for binary input, we first prove the following lemma.

\begin{lemma}
Let $f(x)$ be a function with the following two properties
\begin{enumerate}
\item $f(x)$ is a monotonic (either increasing or decreasing) function defined on ($-\infty, \infty$). \label{p01}
\item $|f'(x)| < 1$ $\forall x$. 
\end{enumerate} 
Then, the following properties hold
\begin{enumerate}
\setcounter{enumi}{2}
\item The equation, $f(x)=x$, has a unique solution.
\item Let $x_s$ be the solution of $f(x) = x$. Let $x_k$ for $k = 1, 2, 3 \cdots$ be a sequence obtained by successively applying $f$ starting from an initial value $x_0$, i.e., $x_k = f(x_{k-1})$. Then, for any $x_0$, $x_k$ approaches to $x_s$ as $k \to \infty$.
\item  Let $g(x) = c\cdot f(x-a)+b$ for some real values $a$, $b$ and $-1 \leq c \leq 1$. The properties 1) to 4) also hold for $g(x)$. \\
\end{enumerate} 
\end{lemma}

\begin{proof}
Property 3) is obvious from 1) and 2). 
Proof of 4) is as follows. Let  $x_{k-1} = x_s + \Delta_{k-1}$ and $x_k = f(x_s+\Delta_{k-1}) = x_s + \Delta_k$. Suppose that $\Delta_{k-1} > 0$. Then, we have
\begin{align}
|\Delta_k| &= \left| \int_{x_s}^{x_s+\Delta_{k-1}} f'(x) {\mathrm d}x \right| \nonumber \\
& \overset{(a)}= \int_{x_s}^{x_s+\Delta_{k-1}} \left| f'(x) \right|{\mathrm d}x  \nonumber \\
& \overset{(b)}< \int_{x_s}^{x_s+\Delta_{k-1}} {\mathrm d}x = \Delta_{k-1} =  |\Delta_{k-1}| \nonumber
\end{align}
where (a) is due to the monotony of $f(x)$, by which $f'(x)$ always has the same sign, and (b) is due to $|f'(x)| < 1 \; \forall x$. Similarly, one also can show the same result, $|\Delta_k| < |\Delta_{k-1}|$, for $\Delta_{k-1} < 0$, which ensures that $x_{\infty} \to x_s$.

In 4), it is obvious that $g(x)$ also has the properties 1) and 2) since shift does not alter the slope and, with $-1\leq c \leq +1$, $|g'(x)| = |cf'(x-a)| \leq 1$. Hence, 3) and 4) also hold for $g(x)$.
\end{proof}

Fig.3(a) and (b) show two examples of the convergence in {\it lemma 2} (Property 4)). From {\it lemma 2}, one can prove the convergence of the forward-backward recursion for binary input as follows.

\begin{theorem}
The forward and backward recursion in (\ref{eq18}) for binary input converges to a unique fixed point as iteration goes to infinity.
\end{theorem}

\begin{proof}
Note that the forward recursion at the $j$th node is of a form 
\begin{align}
l_{j-1 \to j} &= g_j(l_{j-2 \to j-1}) \equiv b_j - \zeta( l_{j-2 \to j-1} - a_j ; c_j ) \nonumber 
\end{align}
for some real values, $a_j$, $b_j$ and $c_j$, where $\zeta(\cdot)$ satisfies the properties 1) to 5) in {\it lemma 2} and, hence, so is $g_j(\cdot)$. Define one iteration as one complete turn of the recursions along the ring, such that the LLR of the $j$th data bit at the $k$th iteration, $l_{j-1 \to j}^{(k)}$, is represented as
\begin{align}
l_{j-1 \to j}^{(k)} &= g_{T,j}(l_{j-1 \to j}^{(k-1)}) \nonumber \\
&= g_j \circ g_{j-1} \circ \dots \circ g_{1} \circ g_{M} \circ \dots g_{j+1}(l_{j-1 \to j}^{(k-1)}) \nonumber 
\end{align}
where, from the chain rule, we have
\begin{align}
\frac{{\mathrm d}g_{T,j}}{{\mathrm d}l} &= \frac{{\mathrm d}g_{j}}{{\mathrm d}g_{j-1}}\cdot
\frac{{\mathrm d}g_{j-1}}{{\mathrm d}g_{j-2}}\cdots \frac{{\mathrm d}g_{1}}{{\mathrm d}g_{M}}\cdot \frac{{\mathrm d}g_{M}}{{\mathrm d}g_{M-1}}\cdots \frac{{\mathrm d}g_{j+1}}{{\mathrm d}l} \nonumber
\end{align}
Since all $g_j(\cdot)$'s satisfy the properties 1) and 2) in {\it lemma 2}, so is $g_{T,j}(\cdot)$. And, hence, from 3) to 5), $l_{j-2 \to j-1}^{(k)}$ converges to a unique fixed point as $k \to \infty$. 

The convergence of the backward recursion can be proved similarly. Note that the two fixed points from the forward and backward recursion do not necessarily equal.
\end{proof}

\section{SINR Analysis via Density Evolution}

In this section, we use the model in Section III-{\it B} to determine the stationary distribution of the LLR, $l_{j-1 \to j}$. Assuming the channel matrix ${\pmb H}$ and the noise power $\sigma^2$ are fixed,
we develop the density evolution of messages between neighboring nodes. In channel coding context, the density evolution in an iterative decoder assumes all-zero sequence is sent and the LLR mean is tracked with the number of iterations, assuming the LLR is Gaussian with its variance being the same as its mean. The density evolution used in \cite{R05} assumes the same, even though the approach differs. 

In this section, we will also assume the LLRs are Gaussian and will track their mean and variance. The differences here from those in iterative channel decoding are that 1) both mean and variance have to be tracked along the ring, where the message of each node has different statistics and, hence, 2) we cannot assume all-zero input since the statistics of the current message, $l_{j-1 \to j}$, depends not only on the background noise but also on the other data. Fortunately, symmetry holds for binary data and the message depends largely on the previous data only so that one can proceed as follows: Under symmetry, we denote the mean and variance of $l_{j-1 \to j}$ as $m_{j|j-1}x_j$ and $v_{j|j-1}$, where both $m_{j|j-1}$ and $v_{j|j-1}$ are non-negative and the mean $m_{j|j-1}x_j$ has the same sign as that of $x_j$. In this definition, $m_{j|j-1}$ can be interpreted as a reliability of $l_{j-1 \to j}$. Then, supposing that $x_j = +1$, we evaluate $(m_{j|j-1},v_{j|j-1})$ for given $(m_{j-1|j-2},v_{j-1|j-2})$.

Assume uniform priors, i.e., $l_{a,j} = 0 \; \forall j$ and suppose that $x_j = +1$. Using (\ref{eq04a}), the forward message passing in (\ref{eq18}) becomes
\begin{align}
l_{j-1 \to j} &= 4y_{j|j-1}^{(R)} - \zeta( l_{j-2 \to j-1} + 2d_{j|j-1} ; 2a_{j|j-1,j-1}^{(R)} )  \nonumber \\ 
&= 4a_{j|j-1,j}^{(R)} + 4n_{j|j-1}^{(R)} + 4a_{j|j-1,j-1}^{(R)}x_{j-1} \nonumber \\
&- \zeta( l_{j-2 \to j-1} + 2d_{j|j-1} ; 2a_{j|j-1,j-1}^{(R)} ) \nonumber \\
&= 4a_{j|j-1,j}^{(R)} + 4n_{j|j-1}^{(R)} + 4a_{j|j-1,j-1}^{(R)}e_{j-1|j-2}(x_{j-1}) \nonumber \\ 
&= 4z_{j|j-1} + 4a_{j|j-1,j-1}^{(R)}e_{j-1|j-2}(x_{j-1}) \label{eq22} 
\end{align}
where, from (\ref{eq19}) to (\ref{eq21}),
\begin{align}
z_{j|j-1} =& a_{j|j-1,j}^{(R)} + n_{j|j-1}^{(R)} \label{eq23} \\
y_{j|j-1}^{(R)} =& a_{j|j-1,j}^{(R)} + n_{j|j-1}^{(R)} + a_{j|j-1,j-1}^{(R)}x_{j-1} \nonumber \\
=& z_{j|j-1} + a_{j|j-1,j-1}^{(R)}x_{j-1} \label{eq24} \\
e_{j-1|j-2}&(x_{j-1}) = x_{j-1} - \frac{\zeta( l_{j-2 \to j-1} + 2d_{j|j-1} ; 2a_{j|j-1,j-1}^{(R)} )}{4a_{j|j-1,j-1}^{(R)}} \label{eq25} 
\end{align}
\begin{align}
d_{j|j-1} =& \frac{2a^{(R)}_{j|j-1,j-1}}{\sigma^2_{j|j-1}} \left( a^{(R)}_{j|j-1,j} + n^{(R)}_{j|j-1}\right)  \nonumber \\
+& \frac{2|a_{j|j-1,j-1}|^2}{\sigma^2_{j|j-1}} x_{j-1} + \frac{2a^{(I)}_{j|j-1,j-1}}{\sigma^2_{j|j-1}} n^{(I)}_{j|j-1} \nonumber \\ 
=& \frac{2a^{(R)}_{j|j-1,j-1}}{\sigma^2_{j|j-1}} z_{j|j-1} + \frac{2|a_{j|j-1,j-1}|^2}{\sigma^2_{j|j-1}} x_{j-1} \nonumber \\
+& \frac{2a^{(I)}_{j|j-1,j-1}}{\sigma^2_{j|j-1}} n^{(I)}_{j|j-1} \label{eq26}
\end{align}
Note that $e_{j-1|j-2}(x_{j-1})$ is the estimation error on $x_{j-1}$ based on the information provided from the current observation, $y_{j|j-1}$ and the previous node message, $l_{j-2 \to j-1}$.

\subsection{Density evolution of the forward-backward recursion}

Here, we determine $(m_{j|j-1},v_{j|j-1})$ for given $(m_{j-1|j-2},v_{j-1|j-2})$ of the previous message, $l_{j-2 \to j-1}$. If the two terms in the last line of (\ref{eq22}) is uncorrelated, the problem is simple. That is, we simply assume the estimation error, $e_{j-1|j-2}(x_{j-1})$, is Gaussian, evaluate its mean and variance and add them to the mean and variance of $z_{j|j-1}$. By iterating such procedure many times, one can obtain the mean and variance of $l_{j-1 \to j}$ (density evolution), even though we need numerical evaluation of integrals. Unfortunately, the two term are correlated due to 1) the inclusion of $z_{j|j-1}$ in the argument of $\zeta(\cdot)$ and 2) the noise + interference in the term $z_{j|j-1}$ in $l_{j-1 \to j}$ and $z_{j-1|j-2}(x_{j-1})$ in $l_{j-2 \to j-1}$ are correlated. (Here, the dependency on $x_{j-1}$ is shown explicitly for $z_{j-1|j-2}$. While, it is not for $z_{j|j-1}$ since we are assuming $x_j = +1$.) One thing that helps make the analysis possible is that they both are well modeled as Gaussian so that the density evolution is numerically tractable by making a few simplifying assumptions.

To this end, we fix $x_{j-1}$ and explore the correlations between the involving variables, of which the randomness solely comes from the noise and other interferences than $x_j$ and $x_{j-1}$. Consider first $z_{j|j-1}$ in the last line of (\ref{eq22}). We have, from (\ref{eq05}) and (\ref{eq06}),
\begin{align}
z_{j|j-1} &\sim {\mathcal N} \left(  a^{(R)}_{j|j-1,j}, {\mathbb E} (n^{(R)}_{j|j-1})^2 \right) \label{eq27}
\end{align}
where, if we assume the suppressed noise + interference, $n_{j|j-1}$, is circularly symmetric, then the variance is given, from (\ref{eq06}) and (\ref{eq07}), by
\begin{align}
{\mathbb E} (n^{(R)}_{j|j-1})^2 &= {\mathbb E} (n^{(I)}_{j|j-1})^2 = \frac{a^{(R)}_{j|j-1,j}}{2}\label{eq28}
\end{align}
This is valid when circularly symmetric constellations, such as QPSK, are used, while it is generally not for non-circularly symmetric real constellations, such as BPSK. Although it has little impact, especially when $M$ is large, it will be certainly more accurate to use the exact variances, which are given in Appendix A.

Now, let us look at the argument of $\zeta(\cdot)$, which can be rewritten as
\begin{align}
l_{j-2 \to j-1}&+ 2d_{j|j-1} =  l_{j-2 \to j-1} + \frac{4a^{(R)}_{j|j-1,j-1}}{\sigma^2_{j|j-1}}  z_{j|j-1}  \nonumber \\
&+ \frac{4|a_{j|j-1,j-1}|^2}{\sigma^2_{j|j-1}} x_{j-1} + \frac{4a^{(I)}_{j|j-1,j-1}}{\sigma^2_{j|j-1}} n^{(I)}_{j|j-1}\label{eq29}
\end{align}
where $l_{j-2 \to j-1}$ is assumed to be Gaussian, of which the mean and variance are provided from the previous node as the pair $(m_{j-1|j-2},v_{j-1|j-2})$, i.e.,
\begin{align}
l_{j-2 \to j-1}  &\sim {\mathcal N} \left(  m_{j-1|j-2}x_{j-1}, v_{j-1|j-2} \right) \label{eq30}
\end{align}
The rest three terms, except for $l_{j-2 \to j-1}$, are assumed to be uncorrelated to each other. Although $z_{j|j-1}$ has a weak correlation with $n^{(I)}_{j|j-1}$  as shown in the Appendix A, we will ignore it for analytical simplicity. Unfortunately, $l_{j-2 \to j-1}$ and $z_{j|j-1}$ has non-negligible covariance. Recalling the definition of $z_{j|j-1}$ and $l_{j-2 \to j-1}$, it stems from the covariance between $n^{(R)}_{j|j-1}$ in $z_{j|j-1}$ and $n^{(R)}_{j-1|j-2}$ in $l_{j-2 \to j-1}$, which is given by

\begin{align}
\sigma_{lz,j} &\equiv {\mathbb E} \left[ l_{j-2 \to j-1} \cdot z_{j|j-1} | x_{j-1}\right] \nonumber \\
&- {\mathbb E} \left[ l_{j-2 \to j-1}  | x_{j-1}\right]\cdot {\mathbb E} \left[ z_{j|j-1} \right] \nonumber \\
&= {\mathbb E} \left[n^{(R)}_{j|j-1} n^{(R)}_{j-1|j-2} \right] + 2a^{(R)}_{j-1|j-2} \cdot {\mathbb E} \left[n^{(R)}_{j|j-1} e_{j-1|j-2} \right] \nonumber \\
&\approx {\mathbb E} \left[n^{(R)}_{j|j-1} n^{(R)}_{j-1|j-2} \right]  \label{eq31}
\end{align}
where we ignored the correlation between $n^{(R)}_{j|j-1}$  and  $e_{j-1|j-2}$. 

Noting that
\begin{align}
{\mathbb E} &\left[n^{(R)}_{j|j-1} n^{(R)}_{j-1|j-2} \right]  +{\mathbb E} \left[n^{(I)}_{j|j-1} n^{(I)}_{j-1|j-2} \right] \nonumber \\
&= \Re \left[n_{j|j-1} n^*_{j-1|j-2} \right] \nonumber \\
&= \Re \left[ {\pmb h}^H_j {\pmb K}^{-1}_{\{j,j-1\}} {\pmb K}_{\{j,j-1,j-2\}} {\pmb K}^{-1}_{\{j-1,j-2\}} {\pmb h}_{j-1}  \right] \label{eq32}
\end{align}
and resorting to the circular symmetry, we have 
\begin{align}
{\mathbb E} \left[n^{(R)}_{j|j-1} n^{(R)}_{j-1|j-2} \right]  ={\mathbb E} \left[n^{(I)}_{j|j-1} n^{(I)}_{j-1|j-2} \right], \nonumber
\end{align}
resulting in
\begin{align}
\sigma_{lz,j} &\approx \frac{1}{2} \Re \left[ {\pmb h}^H_j {\pmb K}^{-1}_{\{j,j-1\}} {\pmb K}_{\{j,j-1,j-2\}} {\pmb K}^{-1}_{\{j-1,j-2\}} {\pmb h}_{j-1}  \right]  \label{eq32}
\end{align}
Without circular symmetry, we also may use similar derivation to ${\mathbb E} ( n^{(R)}_{j|j-1})^2$  as shown in the Appendix A.

	According to the argument below, one can rewrite (\ref{eq29}) using two correlated random variables, say $z_{j|j-1}$ and $w_{j|j-1}$, (conditioned on $x_{j-1}$), i.e.,
\begin{align}
l&_{j-2 \to j-1}+ 2d_{j|j-1} \nonumber \\ 
&\overset{\circ} =  \frac{4a^{(R)}_{j|j-1,j-1}}{\sigma^2_{j|j-1}}  z_{j|j-1}  
+ \frac{4|a_{j|j-1,j-1}|^2}{\sigma^2_{j|j-1}} x_{j-1} + w_{j|j-1} \label{eq33}
\end{align}
where $\overset{\circ} =$ represents equivalence in distribution and
\begin{align}
w_{j|j-1} &\equiv   l_{j-2 \to j-1} +  \frac{4a^{(I)}_{j|j-1,j-1}}{\sigma^2_{j|j-1}} n^{(I)}_{j|j-1} \nonumber \\
& \sim {\mathcal N} \left( m_{j-1|j-2} x_{j-1}, v_{j-1|j-2} + \alpha  \right) \label{eq34}
\end{align}
with
\begin{align}
\alpha = \left( \frac{4a^{(I)}_{j|j-1,j-1}}{\sigma^2_{j|j-1}} \right)^2 {\mathbb E}  \left( n^{(I)}_{j|j-1}\right)^2 \label{eq34}
\end{align}
Recalling the covariance between $z_{j|j-1}$ and $w_{j|j-1}$ given by (\ref{eq31}), the mean and covariance of the Gaussian random pair $(z_{j|j-1}, w_{j|j-1})^T$ are given respectively by
\begin{align}
{\pmb \mu}_{j|j-1} &= \begin{bmatrix}
       a^{(R)}_{j|j-1,j} \\[0.3em]
       m_{j-1|j-2} x_{j-1} 
     \end{bmatrix} \label{eq35} \\
{\pmb C}_{j|j-1} &= \begin{bmatrix}
       {\mathbb E}  \left( n^{(R)}_{j|j-1} \right)^2 & \sigma_{lz,j} \\[0.3em]
       \sigma_{lz,j} & v_{j-1|j-2} + \alpha  
     \end{bmatrix} \label{eq36}
\end{align}
Now, the density evolution can be numerically evaluated by taking average over all possible triple $(x_{j-1}, z_{j|j-1}, w_{j|j-1})$, i.e.,
\begin{align}
m_{j|j-1} &= {\mathbb E}_{x,z,w} \left[ L_{j|j-1} (x,z,w) \right]  \label{eq37} \\
v_{j|j-1} &= {\mathbb E}_{x,z,w} \left[  L_{j|j-1} (x,z,w)^2 \right] - m_{j|j-1}^2  \label{eq38}
\end{align}
where
\begin{align}
L&_{j|j-1} (x,z,w) = 4z + 4a^{(R)}_{j|j-1,j-1} x  \nonumber \\
& -\zeta \left(  \frac{4a^{(I)}_{j|j-1,j-1}}{\sigma^2_{j|j-1}} z + \frac{4|a_{j|j-1,j-1}|^2}{\sigma^2_{j|j-1}} x + w ; a^{(R)}_{j|j-1,j-1} \right) \label{eq39} \\
{\mathbb E}&_{x,z,w} \left[  g(x,z,w) \right] =  \nonumber \\
&\sum_{x=\pm 1} \frac{1}{2} \displaystyle\iint_{-\infty}^{\infty} g(x,z,w) \Phi \left( (w,z)^T; {\pmb \mu}_{j|j-1}, {\pmb C}_{j|j-1}  \right) {\mathrm d}w {\mathrm d}z \label{eq40}
\end{align}
and $\Phi ( (w,z)^T; {\pmb \mu}, {\pmb C}  )$ is bi-variate Gaussian density function with mean ${\pmb \mu}$ and covariance matrix ${\pmb C}$.

\subsection{SINR of the final belief}

The final belief is given by ${\hat l}_j = l_{j-1 \to j} + l_{j+1 \to j}$, for which the mean and variance are given by
\begin{align}
{\mathbb E}[  {\hat l}_j ] &= m_{j|j-1} + m_{j|j+1} \label{eq41} \\
{\mathbb V \rm{ar}}[  {\hat l}_j ] &= v_{j|j-1} + v_{j|j+1} + 2{\mathbb E}[ l_{j-1 \to j} l_{j+1 \to j} ] \nonumber \\
&-  2m_{j|j-1} m_{j|j+1} \label{eq42}
\end{align}
Since the two terms $z_{j|j \pm 1}$ and $e_{j\pm 1|j \pm 2}$ in $l_{j \pm 1 \to j}$ have correlation to each other, the computation of ${\mathbb E}[ l_{j-1 \to j} l_{j+1 \to j} ]$ in (\ref{eq42}) is little bit tricky. However, noting that $e_{j+1|j+2}$ and $e_{j-1|j-2}$ are conditionally independent for given $z_{j|j + 1}$ and $z_{j|j - 1}$, it can be obtained as follows.
\begin{align}
{\mathbb E}[ l_{j-1 \to j} l_{j+1 \to j} ] &=  \\
\iint {\mathbb E}[ l_{j-1 \to j} |z_f] & \cdot {\mathbb E}[  l_{j+1 \to j} |z_b] 
\Phi \left( (z_f,z_b)^T; {\pmb \mu}_{z,j}, {\pmb C}_{zz,j}  \right){ \mathrm d} z_f {\mathrm d} z_b \label{eq43}
\end{align}
where
\begin{align}
{\pmb \mu}_{z,j} &= \begin{bmatrix}
       a^{(R)}_{j|j-1,j} \\[0.3em]
       a^{(R)}_{j|j+1,j} 
     \end{bmatrix} \label{eq44} \\
{\pmb C}_{zz,j} &= \begin{bmatrix}
       {\mathbb E}  \left( n^{(R)}_{j|j-1} \right)^2 & \sigma_{zz,j} \\[0.3em]
       \sigma_{zz,j} & {\mathbb E}  \left( n^{(R)}_{j|j+1} \right)^2  
     \end{bmatrix} \label{eq45} \\
\sigma_{zz,j} &\approx \frac{1}{2} \Re \left[ {\pmb h}^H_j {\pmb K}^{-1}_{\{j,j-1\}} {\pmb K}_{\{j+1,j,j-1\}} {\pmb K}^{-1}_{\{j,j+1\}} {\pmb h}_{j}  \right] \label{eq46}
\end{align}
\begin{align}
{\mathbb E}[ l_{j\pm1 \to j} |z] &= {\mathbb E}_{w,x} \left( L_{j|j\pm 1} (x,z,w) \right) \label{eq47}
\end{align}
with $L_{j|j\pm 1} (x,z,w)$ given by (\ref{eq39}). Note that in (\ref{eq46}) we assumed circular symmetry as in (\ref{eq32}) and (\ref{eq47}) can be obtained similarly to (\ref{eq37}). Using these, the SINR for the $j$th data symbol is given by
\begin{align}
\gamma_j &= \frac{\left( {\mathbb E}[  {\hat l}_j ] \right)^2 }{{\mathbb V \rm{ar}}[  {\hat l}_j ]} \label{eq48}
\end{align}

\begin{figure*}[!tH]
  \centerline{\resizebox{1.7\columnwidth}{!}{\includegraphics{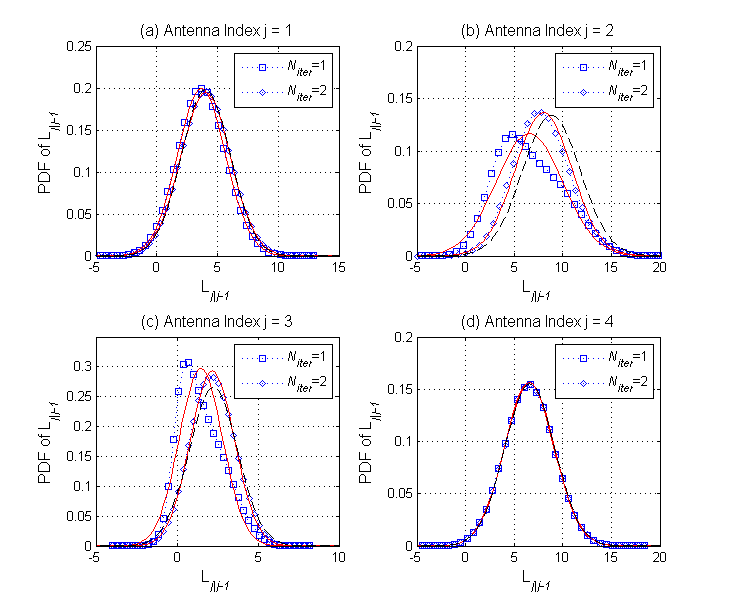}}}
\caption{Density Evolution Example for $N = M = 4$. Dotted with marks: Measured from simulation result, Solid line: Estimated via DE, Dashed: Distribution under perfect cancellation. After two iterations, no changes in LLR density have been observed in all antennas. SNR = $1/\sigma^2$ = 6dB.}
   \label{Fig04}
\end{figure*}

It will also be quite interesting to consider an upper bound on SINR, which is easy to obtain while gives a quite tight bound on bit error rate. Ignoring the estimation error on the previous/next variable, the upper bound is given by
\begin{align}
\gamma_j &\leq \gamma_{bound,j} = \frac{\left( a^{(R)}_{j|j-1,j} + a^{(R)}_{j|j+1,j} \right)^2}{ {\mathbb E}  \left( n^{(R)}_{j|j-1} \right)^2 + {\mathbb E}  \left( n^{(R)}_{j|j+1} \right)^2 + 2\sigma_{zz,j}} \label{eq49}
\end{align}

\section{Numerical Results}

\subsection{Density Evolution Example}

First, we test how well the analysis work for a fixed channel matrix in micro scopic point of view. The channel matrix used is shown below, which was obtained by a random generation and rounding each element below one tenth.
\begin{align}
&{\pmb H}_{ex} = \nonumber \\ 
&\begin{bmatrix}
       -0.1-0.1j & -0.5 & -0.4-0.1j & -0.2+0.8j  \\[0.3em]
       +0.2-0.7j & -0.2+0.2j & -0.1+0.2j & -0.1-0.1j  \\[0.3em]
       -0.1-0.1j & +0.2+0.8j & -0.4-0.2j & 0.4j  \\[0.3em]
       +0.1-0.4j & -0.4+0.2j & +0.2+0.5j & +0.2-0.3j  
     \end{bmatrix} \nonumber
\end{align}
Using this channel, we generate pairs of random binary data and Gaussian noise vector many times. With each pair, we apply the forward-backward recursion to obtain the LLR’s and measured its empirical density. When applying the algorithm, we performed the message passing in parallel fashion, i.e., all the messages, $l_{j-1 \to j} \; \forall j$, are initialized to zero and the message passing begins at the same time for all the nodes. Although it is typical to apply forward and backward recursion sequentially, we applied it in parallel here so that we can observe the density evolution for all the nodes at the same pace. 

The results are shown in Fig.4 (a) to (d), respectively, for each node, where we plotted the density of the forward message for the first two iterations. The dotted lines with marks are the measured density, the solid lines are estimated via density evolution and the dashed lines are the density with perfect cancellation, i.e., $e_{j|j-1}(x_j) = 0 \; \forall j$ such that the SINR is given by (\ref{eq49}). In Fig.4, it can be seen that the messages produced by node 1 and 4 converge almost at the first iteration and we cannot see changing densities with more iteration. On the other hand, one can see explicit density evolution for node 2 and 3, where the densities do not look like Gaussian at the first iteration, while they become more like Gaussian after the second message update, after which we could not observe any further changes (and, hence, we did not plot them). Note also that the densities after convergence show no or a little degradation from those with perfect cancellation of the companion signal. This suggests us that the performance based on the density under the perfect cancellation might provide a tight upper bound on various performance measures.

\begin{figure}[!t]
  \centerline{\resizebox{1.1\columnwidth}{!}{\includegraphics{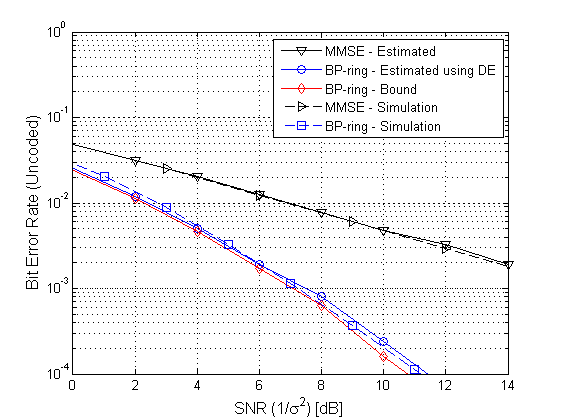}}}
   \caption{Uncoded bit error rate performance of MMSE receiver and BP-Ring in \cite{R13}. Solid line: Estimated, Dashed line: Simulation results.}
   \label{Fig05}
\end{figure}

\subsection{BER performance via Density evolution}

To show the validity of the density evolution approach, we compare the bit error rate (BER) curves estimated using the density evolution with that obtained by simulation results without channel coding. To this end, we generated 1600 random MIMO channel matrices in the same way as did in \cite{R13} and applied density evolution to obtain the SINR of the final beliefs for each node. Using these SINRs, the BER and its bound are estimated by

\begin{align}
P_E &= {\mathbb E}_{\pmb H} \left[  \frac{1}{M} \sum_{j=1}^{M} Q \left( \sqrt{\gamma_j}  \right) \right] 
\geq {\mathbb E}_{\pmb H} \left[  \frac{1}{M} \sum_{j=1}^{M} Q \left( \sqrt{\gamma_{j,{\textrm bound}}}  \right) \right] \label{eq48}
\end{align}
where ${\mathbb E}_{\pmb H}[\cdot]$  is the average over all channel matrices generated. We set the number of transmit and receive antenna to 4 and the number of inner iterations to 2. Fig.5 compares the BER curves obtained via density evolution and its lower bound obtained by the SINR bound in (\ref{eq49}) with that obtained by simulation. The simulation results are the same as those in \cite{R13} without channel coding. The figure shows that (a) the BER obtained via density evolution matches perfectly to the simulation results and (b) the BER bound is quite tight showing the estimation error from the previous/next nodes has negligible effect especially for binary input. Fig.6 shows the SINR averaged over the same set of random channels, where one can see approximately 1.7 dB SINR gain over the linear MMSE receivers at low SNR region (at around 2dB SNR), which is quite good match to the result in \cite{R30}.

\begin{figure}[!t]
  \centerline{\resizebox{1.1\columnwidth}{!}{\includegraphics{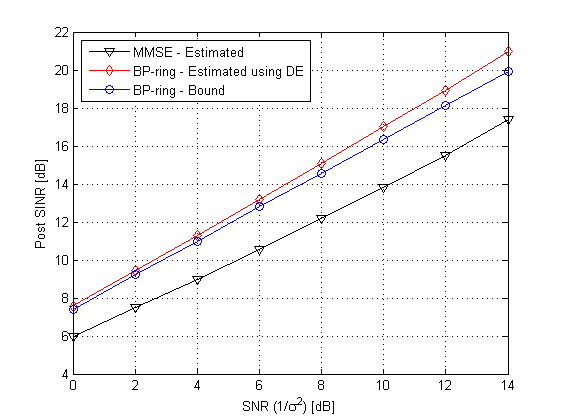}}}
   \caption{Average SINR of MMSE receiver and BP-Ring in \cite{R13}.}
   \label{Fig06}
\end{figure}

\begin{figure*}
\begin{align}
{\mathbb E}  \left( n^{(R)}_{j|j-1}n^{(I)}_{j|j-1} \right) &= 
{\mathbb E}  \left( \Re \left[{{\pmb w}_j^{H}}({\pmb H}_{\{j,j-1\}} {\pmb x}_{\{j,j-1\}}+ {\pmb n}) \right] \Im \left[{{\pmb w}_j^{H}} ({\pmb H}_{\{j,j-1\}} {\pmb x}_{\{j,j-1\}}+ {\pmb n}) \right]\right) \nonumber \\
&= {{\pmb w}_j^{(R)}}^T \left( {\pmb H}^{(R)}_{\{j,j-1\}} {{{\pmb H}^{(I)}}^T_{\{j,j-1\}}} \right) {{\pmb w}_j^{(R)}} 
- {{\pmb w}_j^{(R)}}^T \left( {\pmb H}^{(R)}_{\{j,j-1\}} {{{\pmb H}^{(R)}}^T_{\{j,j-1\}}}  + \frac{\sigma^2}{2} {\pmb I} \right) {{\pmb w}_j^{(I)}}\nonumber \\
&+ {{\pmb w}_j^{(I)}}^T \left( {\pmb H}^{(I)}_{\{j,j-1\}} {{{\pmb H}^{(R)}}^T_{\{j,j-1\}}}  \right) {{\pmb w}_j^{(I)}} 
- {{\pmb w}_j^{(I)}}^T \left( {\pmb H}^{(I)}_{\{j,j-1\}} {{{\pmb H}^{(I)}}^T_{\{j,j-1\}}} + \frac{\sigma^2}{2} {\pmb I} \right) {{\pmb w}_j^{(R)}} \nonumber
\end{align}
\hrulefill
\end{figure*}

\section{Concluding Remarks}

In this paper, we considered the convergence and density evolution of a low complexity MIMO detection based on belief propagation over ring-type pair-wise graph. The algorithm was originally proposed in \cite{R13}, where utilizing the ring-type pair-wise graph the belief propagation algorithm could be expressed as a forward backward recursion and the link level performance and the convergence for Gaussian input have been analyzed. In this paper, we extended the convergence analysis to discrete alphabet. Specifically, we proved the convergence of the forward-backward recursion and devised a density evolution approach to provide an asymptotic performance in terms of average bit error rate and SINR. The BER curves shows perfect match with simulation results provided in \cite{R13}, which validates the density evolution approach for binary input and the performance improvements of the algorithm in \cite{R13} over the linear MMSE receiver.

\begin{appendix}[Statistics of $n^{(R)}_{j|j\pm 1}$ and $n^{(I)}_{j|j\pm1}$ ]

Without loss of generality, we consider only $n^{(R)}_{j|j-1}$ and $n^{(I)}_{j|j-1}$. First, from the definition of $n_{j|j-1}$ 
, we obtain 

\begin{align}
{\mathbb E} \left( n^{(R)}_{j|j-1} \right)^2 &= {\mathbb E}  \left( \Re \left[{{\pmb w}_j^{H}} \cdot ({\pmb H}_{\{j,j-1\}} {\pmb x}_{\{j,j-1\}} {\pmb n}) \right] \right)^2 \nonumber \\
&= \; {{\pmb w}_j^{(R)}}^T \left( {\pmb H}^{(R)}_{\{j,j-1\}} {{{\pmb H}^{(R)}}^T_{\{j,j-1\}}} + \frac{\sigma^2}{2} {\pmb I} \right) {{\pmb w}_j^{(R)}} \nonumber \\
&+ {{\pmb w}_j^{(R)}}^T \left( {\pmb H}^{(R)}_{\{j,j-1\}} {{{\pmb H}^{(I)}}^T_{\{j,j-1\}}}  \right) {{\pmb w}_j^{(I)}} \nonumber \\
&+ {{\pmb w}_j^{(I)}}^T \left( {\pmb H}^{(I)}_{\{j,j-1\}} {{{\pmb H}^{(R)}}^T_{\{j,j-1\}}}  \right) {{\pmb w}_j^{(R)}} \nonumber \\
&+ {{\pmb w}_j^{(I)}}^T \left( {\pmb H}^{(I)}_{\{j,j-1\}} {{{\pmb H}^{(I)}}^T_{\{j,j-1\}}} + \frac{\sigma^2}{2} {\pmb I} \right) {{\pmb w}_j^{(I)}} \nonumber 
\end{align}
where we used ${\mathbb E}[{\pmb {xx}}^H] = {\pmb I}$ and
\begin{align}
{\mathbb E} \left[ {\pmb n}^{(R)}{{\pmb n}^{(R)}}^T \right] &= {\mathbb E} \left[ {\pmb n}^{(R)}{{\pmb n}^{(R)}}^T \right] = \frac{\sigma^2}{2} {\pmb I} \nonumber \\
{\mathbb E} \left[ {\pmb n}^{(R)}{{\pmb n}^{(I)}}^T \right] &= {\pmb 0} \nonumber
\end{align}
Note that, from independence among columns of $\pmb H$ and between the real and imaginary parts, ${\pmb w}_j^{(R)}$ and ${\pmb w}_j^{(I)}$ are quasi-orthogonal (i.e., ${{\pmb w}_j^{(R)}}^T{\pmb w}_j^{(I)} \approx {\pmb 0}$) and, hence, the covariance is weak especially when $N$ is large. And, if this is the case, we can approximate
\begin{align}
{\mathbb E} \left( n^{(R)}_{j|j-1} \right)^2 &\approx {{\pmb w}_j^{(R)}}^T \left( {\pmb H}^{(R)}_{\{j,j-1\}} {{{\pmb H}^{(R)}}^T_{\{j,j-1\}}} + \frac{\sigma^2}{2} {\pmb I} \right) {{\pmb w}_j^{(R)}} \nonumber \\
&+ {{\pmb w}_j^{(I)}}^T \left( {\pmb H}^{(I)}_{\{j,j-1\}} {{{\pmb H}^{(I)}}^T_{\{j,j-1\}}} + \frac{\sigma^2}{2} {\pmb I} \right) {{\pmb w}_j^{(I)}} \nonumber
\end{align}
Similarly, we obtain
\begin{align}
{\mathbb E} \left( n^{(I)}_{j|j-1} \right)^2 &\approx {{\pmb w}_j^{(R)}}^T \left( {\pmb H}^{(I)}_{\{j,j-1\}} {{{\pmb H}^{(I)}}^T_{\{j,j-1\}}} + \frac{\sigma^2}{2} {\pmb I} \right) {{\pmb w}_j^{(R)}} \nonumber \\
&+ {{\pmb w}_j^{(I)}}^T \left( {\pmb H}^{(R)}_{\{j,j-1\}} {{{\pmb H}^{(R)}}^T_{\{j,j-1\}}} + \frac{\sigma^2}{2} {\pmb I} \right) {{\pmb w}_j^{(I)}} \nonumber
\end{align}

The covariance between $n^{(R)}_{j|j-1}$ and $n^{(I)}_{j|j-1}$ can also be obtained similar way, as shown on top of this page.
From independence among columns of $\pmb H$ and between the real part and imaginary part, ${{\pmb w}_j^{(R)}}^T{\pmb w}_j^{(I)} \approx {\pmb 0}$ and ${{\pmb H}_{j,j-1}^{(R)}}{{\pmb H}_{j,j-1}^{(I)}}^T \approx {\pmb 0}$, so that the covariance is weak especially when $N$ is large.

\end{appendix}

\renewcommand{\baselinestretch}{1.0}
\bibliographystyle{IEEEbib}

\bibliography{references_combined}

\end{document}